\documentclass[10pt,twocolumn,twoside]{IEEEtran}
\usepackage{cite}
\usepackage{amsmath}
\usepackage{times}
\usepackage{latexsym}
\usepackage{amssymb}
\usepackage{setspace}
\usepackage{epsfig}
\usepackage{balance}
\usepackage{amssymb}
\usepackage[center]{caption2}
\usepackage{stfloats}
\usepackage{cases}
\usepackage{array}
\usepackage{pifont}
\usepackage{setspace}
\usepackage{fancyhdr}
\usepackage{amscd}
\usepackage{algorithm}
\usepackage{algorithmic}
\usepackage{pst-all}
\usepackage{pstricks}
\usepackage{epsfig}
\usepackage{balance}
\usepackage{graphicx}
\allowdisplaybreaks
\usepackage[T1]{fontenc}
\usepackage[latin1]{inputenc}
\usepackage{array,booktabs,arydshln,xcolor}
\newcommand\VRule[1][\arrayrulewidth]{\vrule width #1}

\allowdisplaybreaks

\newtheorem{lemma}{Lemma}

\IEEEoverridecommandlockouts
\makeatletter
  \def\vhrulefill#1{\leavevmode\leaders\hrule\@height#1\hfill \kern\z@}
\makeatother
\begin{document}
\title{Secure Communications with Cooperative Jamming: Optimal Power Allocation and Secrecy Outage Analysis}
\author{Kanapathippillai Cumanan,~\IEEEmembership{Member, IEEE}, George C. Alexandropoulos,~\IEEEmembership{Senior Member, IEEE},\\
Zhiguo Ding,~\IEEEmembership{Senior Member, IEEE}, and George K. Karagiannidis~\IEEEmembership{Fellow, IEEE} \vspace{-0.3in}
\thanks{K. Cumanan is with the Department of Electronics, University of York, York, YO10 5DD, UK. (Email: kanapathippillai.cumanan@york.ac.uk).}
\thanks{G. C. Alexandropoulos is with the Mathematical and Algorithmic Sciences Lab, France Research Center, Huawei Technologies Co. Ltd., 92100 Boulogne-Billancourt, France (Email: george.alexandropoulos@huawei.com).}
\thanks{Z. Ding is with the School of Computing and Communications, Lancaster University Lancaster, LA1 4WA, UK. (Email: z.ding@lancaster.ac.uk).}
\thanks{G. K. Karagiannidis is with Khalifa University, Abu Dhabi, UAE and with Aristotle University of Thessaloniki, Thessaloniki, Greece. (Email: geokarag@auth.gr).}
%\thanks{The work of K. Cumanan and Z. Ding was supported by H2020-MSCARISE-2015 under grant number 690750. The work of Z. Ding was supported by the U.K. EPSRC under grant number EP/L025272/1.}
}
\maketitle
\IEEEpeerreviewmaketitle
\begin{abstract}
This paper studies the secrecy rate maximization problem of a  secure wireless communication system, in the presence of multiple eavesdroppers. The security of the communication link is enhanced through cooperative jamming, with the help of multiple jammers. First, a feasibility condition is derived to achieve a positive secrecy rate at the destination. Then, we solve the original secrecy rate maximization problem, which is not convex in terms of power allocation at the jammers. To circumvent this non-convexity, the achievable secrecy rate is approximated for a given power allocation at the jammers and the approximated problem is formulated into a geometric programming one. Based on this approximation, an iterative algorithm has been developed to obtain the optimal power allocation at the jammers. Next, we provide a bisection approach, based on one-dimensional search, to validate the optimality of the proposed algorithm. In addition, by assuming Rayleigh fading, the secrecy outage probability (SOP) of the proposed cooperative jamming scheme is analyzed. More specifically, a single-integral form expression for SOP is derived for the most general case as well as a closed-form expression for the special case of two cooperative jammers and one eavesdropper. Simulation results have been provided to validate the convergence and the optimality of the proposed algorithm as well as the theoretical derivations of the presented SOP analysis.
\end{abstract}
\linespread{0.89}
\section{Introduction}
\indent {\Huge{P}}hysical (PHY) layer security has recently received considerable attention as a significant candidate to enhance the quality of secure communication in emerging and future wireless networks, including the fifth generation (5G) standard \cite{Renzo_Commun_Mag_J15}. In this new paradigm, the propagation characteristics of wireless channels are exploited against passive eavesdroppers and active attacks through PHY layer secret key generation and authentication schemes, while complementing the conventional cryptographic methods \cite{Zeng_Commun_Mag_J15}. The fundamental concept of information-theoretic security was first investigated in \cite{Wyner_J75} and \cite{Korner_Info_Theory_J78}, where it was shown that secure communication is feasible when the channel quality of legitimate parties is better than that of the eavesdropper. However, in practice, this is not always possible and so, the performance of PHY layer security is limited.\\
\indent In order to circumvent the performance limitations introduced by the unfavourable wireless channel conditions, cooperative jamming has been proposed as an enabler of secrecy communication\cite{Maged_J15,Duong_Comm_Lett_J14,Caijun_Commun_J15,George_KK_TVT_J15,Zou_JSAC_J13,Cuma_TVT_J14,Zheng_TVT_J15,Cuma_JSTSP_J16,Wei_Chen_Wireless_Lett_J15,Zheng_Wireless_Lett_J15}. Under this approach, jamming signals are transmitted to improve the secrecy rate performance, by introducing interference at the eavesdroppers. In \cite{Weber_Sig_Process_J11}, different secrecy rate optimization problems have been solved for relay network based on cooperative jamming, where the relays transmit noise to confound the eavesdroppers. However, these optimization problems have been considered with a total relay power constraint. For the same network, a cooperative jamming scheme has been proposed in \cite{Poor_Sig_Process_J10} with no interference leakage to the legitimate user. Furthermore, in \cite{Ding_Wireless_J11}, opportunistic cooperative jamming and relay chatting schemes have been developed, without the knowledge of eavesdropper channel state information (CSI), and the performance of these schemes have been evaluated through the secrecy outage probability (SOP) criterion. On the other hand, in \cite{Petropulu_Info_security_J13}, an uncoordinated cooperative jamming scheme with multi-antenna relays has been investigated by nulling the interference leakage at the destination and the corresponding SOP has been quantified with eavesdroppers' statistical CSI. In \cite{Gan_CJ_Sig_Process_J11}, optimal cooperative jamming scheme has been proposed with multiple relays in the presence of a single eavesdropper, where the optimal relay coefficients have been obtained through an one-dimensional search scheme.\\
\indent The SOP of a multi-user wireless communication system, that consists of multiple users who transmit to a base station, while multiple eavesdroppers attempt to tap their transmissions, has been analyzed over Rayleigh fading channels in \cite{C:Secrecy_Multiuser}. In \cite{Petropulu_SSP12}, a closed-form expression of SOP was derived for Rayleigh fading channels in a secrecy network with a multi-antenna source and a single-antenna destination in the presence of a single-antenna eavesdropper. Finally, in \cite{J:Maged_Mallik}, the SOP performance of the multiple-input multiple-output (MIMO) wiretap channel, employing transmit antenna selection and receive generalized selection combining, has been analyzed over Nakagami-$m$ fading channels.\\
\indent In this paper, we consider a PHY layer security network with single-antenna nodes, where a source-destination pair establishes secured communication, with the help of multiple jammers in the presence of multiple eavesdroppers. For this network setup, we first present a feasibility condition to achieve a positive secrecy rate at the destination. Then, the secrecy rate maximization problem is solved to determine the optimal power allocation at the jammers, which is a non-convex problem in nature. In order to overcome the non-convexity of the secrecy rate function, we approximate it for a given power allocation at the jammers and formulate the problem into a geometric programming one. Based on this approximation, an iterative algorithm is developed, by updating a better power allocation at each iteration. To validate the optimality of the presented results, we use one-dimensional search based on bisection to determine the optimal power allocation of the original secrecy rate maximization problem. Both the proposed and the one-dimensional search algorithms yield identical results, which confirms the optimality of the proposed algorithm. Moreover, the SOP of the proposed scheme is analyzed over Rayleigh fading channels. A single-integral form expression for the SOP is presented for the most general scenario, whereas a closed-form expression is derived for the special case of two cooperative jammers and one eavesdropper. Finally, numerical and simulation results have been provided to validate the theoretical derivations.\\
\indent The remainder of the paper is organized as follows. The system model and the secrecy rate maximization problem formulation are presented in Section II. A feasibility condition to achieve positive secrecy rate is provided in Section III, whereas Section IV presents an iterative approach for an approximated secrecy rate maximization problem. In Section V, the optimality of the proposed scheme is validated through one-dimensional search. The SOP analysis is derived in Section VI for Rayleigh fading channels, whereas Section VII provides numerical and simulation results to validate the performance of the proposed algorithm and the derived theoretical SOP expressions. Finally, Section VIII concludes this paper.\\
\textbf{Notations:} We use lower-case boldface letters for vectors. $(\cdot)^{T}$ and $|\cdot|$ denote the transpose of a vector and absolute value of a complex number, respectively. $[x]^{+}$ represents $\rm{max}\{x,0\}$ whereas $E\{\cdot\}$, $\rm{Pr}[\cdot]$ and $\nabla(\cdot)$ denote expectation, probability and gradient operator, respectively. The cumulative distribution function (CDF) and the probability density function (PDF) of a random variable (RV) $X$ are represented as $F_X(\cdot)$ and $f_X(\cdot)$, respectively. ${\rm Ei}(\cdot)$ is the exponential integral \cite[eq. (8.211/1)]{B:Gra_Ryz_Book}.
\begin{figure}[t]
\includegraphics[scale = 0.7]{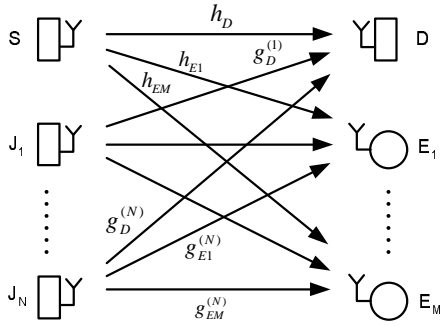}
\centering\caption{The considered secrecy network with one source, one destination and multiple jammers, in the presence of multiple eavesdroppers.}
\label{fig:Mcast_sec_net}
\end{figure}
\section{System Model}
\indent We consider a secrecy network, as shown in Fig. \ref{fig:Mcast_sec_net}, with one source, $S$, which communicates with a destination, $D$ and $N$ cooperative jammers, $J_{1},J_{2},\ldots,J_{N}$, in the presence of $M$ eavesdroppers, $E_{1},E_{2},\ldots,E_{M}$. The source $S$ wishes to transmit secured information to destination $D$. It is assumed that all network nodes are equipped with a single antenna. The channel coefficient between $S$ and $D$ is denoted by $h_{D}$, whereas $h_{Em}$ represents the channel gain between $S$ and the $m^{\rm{th}}$ eavesdropper $E_{m}$, with $m= 1,2,\ldots,M$. In addition, the channel coefficient between the $n^{\rm{th}}$ cooperative jammer $J_{n}$ and $D$ as well as $E_{m}$ are denoted by $g_{D}^{(n)}$ and $g_{Em}^{(n)}$, respectively. The CSI between all nodes are assumed to be perfectly available at $S$, $D$ and $E_{m}~\forall m$. The source $S$ transmits signals to destination $D$ whereas all jammers send interference signals to confound the eavesdroppers.\\
\indent The received signals at $D$ and $E_{m}$ can be mathematically expressed respectively, as
\begin{eqnarray}
  y_{D} &=& \sqrt{P_{s}}h_{D}x_{s}+\sum_{i=1}^{N}\sqrt{P_{i}}g_{D}^{(i)}x_{c}^{(i)}+\eta_{D}\\
  y_{Em} &=& \sqrt{P_{s}}h_{Em}x_{s}+\sum_{i=1}^{N}\sqrt{P_{i}}g_{Em}^{(i)}x_{c}^{(i)}+\eta_{Em}
\end{eqnarray}
where $x_{s}~(\mathbb{E}\{|x_{s}|^2\}= 1)$ and $x_{c}^{(i)}~(\mathbb{E}\{|x_{c}^{(i)}|^2\}=1)$ denote the transmitted signal from $S$ to $D$ and the jamming signal from the $i^{\rm{th}}$ jammer $J_{i}$, respectively. In addition, $\eta_{D}~(\mathbb{E}\{|\eta_{D}|^2\}=\sigma_{D}^{2})$ and $\eta_{Em}~(\mathbb{E}\{|\eta_{Em}|^2\}=\sigma_{Em}^{2})$ represent the noise at node $D$ and $m^{\rm{th}}$ eavesdropper $E_{m}$, respectively. The power allocation at $J_{i}$ and $S$ are denoted by $P_{i}$ and $P_{s}$, respectively. Assuming white Gaussian noise, the achievable secrecy rate at $D$ is defined as
\begin{equation}\label{eq:sec_rate}
  R_{s} = \left[\textrm{log}_{2}\left(1+\gamma_{D}\right)-\textrm{log}_{2}\left(1+\gamma_{{E}_{\rm{max}}}\right)\right]^{+}
\end{equation}
where $\gamma_{E_{\rm{max}}}=\max\left\{\gamma_{E1},\gamma_{E2},\ldots,\gamma_{EM}\right\}$  and $\gamma_{D}$, $\gamma_{Em}$ are the signal-to-interference plus noise ratios (SINR) at $D$ and $E_{m}$, respectively, given by
\begin{eqnarray}
    \gamma_{D} &=& \frac{P_{s}|h_{D}|^{2}}{\sum_{i=1}^{N}P_{i}|g_{D}^{(i)}|^{2}+\sigma_{D}^{2}}\\
    \gamma_{Em} &=& \frac{P_{s}|h_{Em}|^{2}}{\sum_{i=1}^{N}P_{i}|g_{Em}^{(i)}|^{2}+\sigma_{Em}^{2}}.
\end{eqnarray}
\indent For the secrecy network studied in this paper, we consider secrecy rate maximization with transmit power constraint. In particular, we intend to maximize the achievable secrecy rate at the destination node $D$, with the available transmit power at the source node and all $N$ available jammers. The secrecy rate maximization problem can be therefore formulated as
\begin{eqnarray}\label{eq:secrecy_rate_max_v1}
  \textsc{P}1:~~~~\max_{\mathbf{p}\succeq\mathbf{0}} && R_{s}\nonumber\\
   \textrm{s.t.} && P_{i} \leq \bar{P}_{i}, \forall i
\end{eqnarray}
where $\bar{P}_{i}$ is the maximum available transmit power at $J_{i}$ and $\mathbf{p}=[P_{1}\,P_{2}\,\cdots\,P_{N}]^{T}$.
\section{Feasibility Conditions\\ for Positive Secrecy Rate}
\indent The optimization problem $\textsc{P}1$, formulated in \eqref{eq:secrecy_rate_max_v1}, is valid or worth to solve only when it is possible to achieve a positive secrecy rate for a given set of channels and transmit powers at $D$ and $J_{i}$s. Through verifying these feasibility conditions, the source can make a decision whether to solve the secrecy rate maximization to obtain a positive secrecy rate at the destination. Hence, we first investigate the feasibility conditions. From \eqref{eq:sec_rate}, the following conditions need to be satisfied for $m=1,2,\ldots,M$:
\begin{equation}\label{feasible_cond1}
\frac{P_{s}|h_{D}|^2}{\sum_{i=1}^{N}P_{i}|g_{D}^{(i)}|^2+\sigma_{D}^{2}}> \frac{P_{s}|h_{Em}|^2}{\sum_{i=1}^{N}P_{i}|g_{Em}^{(i)}|^2+\sigma_{Em}^{2}}
\end{equation}
By arranging the terms in \eqref{feasible_cond1}, the following equality needs to hold $\forall m$:
\begin{eqnarray}
   |h_{D}|^2\!\!\left(\sum_{i=1}^{N}\!P_{i}|g_{Em}^{(i)}|^2\!+\!\sigma_{Em}^{2}\!\!\right)\!\!\!\!\!\!&>&\!\!\!\!\!\!|h_{Em}|^2\!\!\left(\sum_{i=1}^{N}\!P_{i}|g_{D}^{(i)}|^2\!+\!\sigma_{D}^{2}\!\!\right)\nonumber
\end{eqnarray}
which can be expressed as
\begin{eqnarray}\label{eq:feasible_conditions}
    \mathbf{p}^{T}\left(|h_{D}|^2\mathbf{g}_{Em}\!-\!|h_{Em}|^2\mathbf{g}_{D}\right)\!\!\!\!\!&>&\!\!\!\!\!|h_{Em}|^2\sigma_{D}^{2}\!-\!|h_{D}|^2\sigma_{Em}^{2}
\end{eqnarray}
where
\begin{eqnarray}\label{}
\mathbf{g}_{Em}&=&
\left[
  \begin{array}{cccc}
    |g_{Em}^{(1)}|^2 & |g_{Em}^{(2)}|^2 & \cdots &|g_{Em}^{(N)}|^2  \\
  \end{array}
\right]^{T}\nonumber\\
\mathbf{g}_{D} &=&\left[
                           \begin{array}{cccc}
                            |g_{D}^{(1)}|^2& |g_{D}^{(2)}|^2 & \cdots& |g_{D}^{(N)}|^2\\
                           \end{array}
                         \right]^{T}
\end{eqnarray}
The feasibility conditions given by \eqref{eq:feasible_conditions} can be formulated into the following linear programming problem\cite{boyd_B04}:
\begin{eqnarray} \label{eq:feasibility_LP}
  \min_{\mathbf{p}\succeq0}\!\!\!&&\!\!\!\!\!\mathbf{1}^{T}\mathbf{p}\nonumber\\
  \textrm{s.t.}\!\!\!&&\!\!\!\!\!\mathbf{p}^{T}\left(|h_{D}|^2\mathbf{g}_{Em}\!-\!|h_{Em}|^2\mathbf{g}_{D}\right)\!\!>\!\!|h_{Em}|^2\sigma_{D}^{2}\!-\!|h_{D}|^2\sigma_{Em}^{2},\nonumber\\
  && ~~~~~~~~~~~~~~~~~~~~~~~~~~~~~~~~~~~~~~~~~~~~~~~~~~~\forall m.
\end{eqnarray}
The above convex problem can be easily solved using existing convex optimization software \cite{boyd_B04,Ye_B97}. A positive secrecy rate can be only achieved at the destination node, if the problem in \eqref{eq:feasibility_LP} is feasible. In the following section, we solve the secrecy rate maximization problem, with the assumption that a positive secrecy rate is achievable.
\section{An Iterative Approach for the solution of the secrecy rate maximization problem}
\indent The secrecy rate maximization problem $\textsc{P}1$ given by \eqref{eq:secrecy_rate_max_v1} is non-convex due to the non-convex secrecy rate function and therefore it is challenging to obtain the optimal solution. In this section, we develop an iterative algorithm for the power allocation $\mathbf{p}$ at the jammer nodes, that is based on an approximation to the original problem $\textsc{P}1$. By reformulating \eqref{eq:secrecy_rate_max_v1} and introducing a new slack variable $\tau$, the original secrecy maximization problem $\textsc{P}1$ can be written as
\begin{eqnarray}\label{eq:secrecy_rate_max_v2}
  \textsc{P}2:~~~~ \min_{\mathbf{p}\succeq\mathbf{0},\tau \geq0} && \tau\nonumber\\
   \textrm{s.t.}&& \Gamma_{Em}(\mathbf{p}) \triangleq \frac{\Phi_{Em}(\mathbf{p})}{\Psi_{Em}(\mathbf{p})}\leq\tau,\forall m\nonumber\\ \label{eq:constr_sec_max_v2}
   && P_{i} \leq \bar{P}_{i}, \forall i.
\end{eqnarray}
where
\begin{eqnarray}\label{eq:approx_v1}
\small
   \Psi_{Em}(\mathbf{p})\!\!\!\!\!&\triangleq&\!\!\!\!\!\left(\sum_{i=1}^{N}P_{i}|g_{D}^{(i)}|^2+P_{s}|h_{D}|^2+\sigma_{D}^{2}\right)\nonumber\\
   &&\!\!\!\!\times\left(\sum_{i=1}^{N}P_{i}|g_{Em}^{(i)}|^2+\sigma_{Em}^2\right)\!\!\triangleq\!\!\sum_{k}\psi_{Em}^{(k)}
\end{eqnarray}
   and
\begin{eqnarray}
   \Phi_{Em}(\mathbf{p})\!\!\!\!&\triangleq&\!\!\!\! \left(\sum_{i=1}^{N}P_{i}|g_{Em}^{(i)}|^2+\sigma_{Em}^2+P_{s}|h_{Em}|^{2}\right)\nonumber\\
   &&\!\times\!\left(\sum_{i=1}^{N}P_{i}|g_{D}^{(i)}|^2\!+\!\sigma_{D}^{2}\right).
\end{eqnarray}
\indent In \eqref{eq:approx_v1}, $\psi_{Em}^{(k)}$ represents the individual term in the summation, obtained by expanding function $\Psi_{Em}(\mathbf{p})$. The constraint in \eqref{eq:constr_sec_max_v2} is a quadratic fractional non-convex function. However, the problem in \eqref{eq:secrecy_rate_max_v2} can be converted into a series of geometric programming problems by exploiting the \emph{single condensation method} \cite{Julian_Wireless_Commun_J07}. A fractional constraint with a posynomial numerator and a monomial denominator is convex. The idea of approximating the denominator posynomial with a monomial was presented in \cite{boyd_B04} in order to convert the aforementioned constraint to a convex one. We hereinafter adopt this idea and we approximate $\Psi_{Em}(\mathbf{p})$ (i.e., denominator of the constraint in \eqref{eq:constr_sec_max_v2}) to the best monomial, for a given set of $\mathbf{p}$. The following lemma is required:

\begin{lemma}\label{lemma:Mono_approx}
%Let $g(\mathbf{x})$, be a posynomial defined as
%\begin{equation}\label{}
%    g(\mathbf{x}) = \sum_{k=1}^{K}w_{k}(\mathbf{x}) =\sum_{k=1}^{K}c_{k}x_{1}^{n_{1k}}x_{2}^{n_{2k}}\cdots x_{m}^{n_{mk}}
%\end{equation}
%where $c_{k}$ and $n_{lk}$ with $k=1,2,\ldots,K$ and $l=1,2,\ldots,m$ are positive constants and arbitrary real numbers, respectively.
For a posynomial $g(\mathbf{x})$, the following inequality holds:
\begin{equation}\label{eq:equality_approx}
    g(\mathbf{x}) = \sum_{k=1}^{K}w_{k}(\mathbf{x}) \geq \hat{g}(\mathbf{\hat{x}}) = \prod_{k=1}^{K}\left[\frac{w_{k}(\mathbf{x})}{a_{k}}\right]^{a_{k}}
\end{equation}
where $a_{k}> 0$ and $\sum_{k=1}^{K}a_{k}=1$. Notation $\hat{g}(\mathbf{\hat{x}})$ represents the best approximation of $g(\mathbf{\hat{x}})$ at $\mathbf{\hat{x}}$ with $a_{k} = w_{k}(\mathbf{\hat{x}})/g(\mathbf{\hat{x}})$, and the inequality in \eqref{eq:equality_approx} holds with an equality at this point.
\end{lemma}

\begin{proof}
The proof is provided in Appendix~\ref{App:Lemma1}.
\end{proof}

Based on Lemma \ref{lemma:Mono_approx}, the denominator polynomial function $\Psi_{Em}(\mathbf{p})$ in \eqref{eq:constr_sec_max_v2}, can be approximated as $\hat{\Psi}_{Em}(\mathbf{p})$
\begin{equation}\label{eq:approx_cal}
  \Psi_{Em}(\mathbf{p}) \approx \hat{\Psi}_{Em}(\mathbf{p}) \triangleq \prod_{k=1}^{K}\left[\frac{\psi_{Em}^{(k)}}{\alpha_{k}^{(m)}}\right]^{\alpha_{k}^{(m)}}
\end{equation}
where
\begin{equation}\label{eq:alpha_cal}
  \alpha_{k}^{(m)} \triangleq  \frac{\psi_{Em}^{(k)}}{\hat{\Psi}_{Em}(\mathbf{p})}~~\forall k.
\end{equation}
\indent Using the approximation given by \eqref{eq:approx_cal}, the problem $\textsc{P}2$ can be reformulated for a given set of power allocation $\mathbf{p}$ as
\begin{eqnarray}\label{eq:secrecy_rate_max_v3}
  \textsc{P}3:~~ \min_{\mathbf{p}\succeq\mathbf{0},\tau \geq0} && \tau\nonumber\\
   \textrm{s.t.}&& \hat{\Gamma}_{Em}(\mathbf{p}) \triangleq \frac{\Phi_{Em}(\mathbf{p})}{\hat{\Psi}_{Em}(\mathbf{p})}\leq\tau,~\forall m,\nonumber\\ \label{eq:constr_sec_max_v3}
   && P_{i} \leq \bar{P}_{i}, \forall i.
\end{eqnarray}
\indent The above optimization problem $\textsc{P}3$, which is an approximation of the original $\textsc{P}1$, can be now formulated into a standard geometric programming one. The iterative algorithm A is developed for $\textsc{P}3$, where the power allocation $\mathbf{p}$ is updated at each iteration.
{\renewcommand\baselinestretch{1}\selectfont
\begin{figure}[t]
    \hspace{0em}\hrulefill

 \hspace{0em} {\bf Algorithm A:} Secrecy Rate Maximization

     \hspace{0em}\hrulefill

     \hspace{0em} Step 1: Initialization of power allocation vector $\mathbf{p}$

     \hspace{0em} Step 2: Repeat
\begin{enumerate}
   \item Calculate $\Psi_{Em}(\mathbf{p}),~\forall m$ using \eqref{eq:approx_v1}.
   \item Calculate $\alpha_{k}^{(m)},~\forall k,~m$ using \eqref{eq:alpha_cal}.
   \item Determine $\hat{\Psi}_{Em}(\mathbf{p}),~\forall m$ by using \eqref{eq:approx_cal}.
   \item Solve the standard geometric programming problem in \eqref{eq:secrecy_rate_max_v3}.
\end{enumerate}
\hspace{-0em} Step 3: Until required accuracy is achieved or the maximum number of iterations is reached.

\hspace{-1em} \hrulefill
\end{figure}\par}
The solution of the proposed Algorithm A satisfies the Karush-Kuhn-Tucker (KKT) conditions. This can be validated by proving the following three conditions \cite{Wright1978convex_approx}:
\begin{enumerate}\label{KKT_codtn}
\item $\Gamma_{Em}(\mathbf{p})\leq \hat{\Gamma}_{Em}(\mathbf{p}),~\forall~m,\mathbf{p}$, where $\Gamma_{Em}(\mathbf{p}) = \frac{\Phi_{Em}(\mathbf{p})}{\Psi_{Em}(\mathbf{p})}$.
\item $\Gamma_{Em}(\mathbf{\tilde{p}}) = \hat{\Gamma}_{Em}(\mathbf{\tilde{p}}),~\forall~m,$ where $\mathbf{\tilde{p}}$ denotes the power allocation obtained from the previous iteration of Algorithm A.
\item  $\nabla\Gamma_{Em}(\mathbf{\tilde{p}}) =\nabla\hat{\Gamma}_{Em}(\mathbf{\tilde{p}}),\forall~m$.
\end{enumerate}
The first condition holds due to the fact that $\Psi_{Em}(\mathbf{p})\leq\hat{\Psi}_{Em}(\mathbf{p})$, which is true from Lemma \ref{lemma:Mono_approx}. In addition, the second condition is satisfied from the equality condition in Lemma \ref{lemma:Mono_approx}. The third condition can be validated through proving $\nabla\hat{\Psi}_{Em}(\mathbf{\tilde{p}})=\nabla\Psi_{Em}(\mathbf{\tilde{p}})$ for all $m$:
\begin{eqnarray}
  \nabla\hat{\Psi}_{Em}(\mathbf{\tilde{p}}\!)\!\!\!\!\!&=&\!\!\!\!\!\left[\!\!\frac{\partial\hat{\Psi}_{Em}(\mathbf{\tilde{p}})}{\partial P_{1}}\bigg|_{\tilde{P}_{1}}\!\!\!\frac{\partial\hat{\Psi}_{Em}(\mathbf{\tilde{p}})}{\partial P_{2}}\bigg|_{\tilde{P}_{2}}\!\!\!\!\!\cdots\!\frac{\partial\hat{\Psi}_{Em}(\mathbf{\tilde{p}})}{\partial P_{N}}\bigg|_{\tilde{P}_{N}}\!\right]\!,\nonumber\\
  &&~~~~~~~~~~~~~~~~~~~~~~~~~~~~~~~~~~~~~~~\forall m, \label{eq:partial_diff}\\
   \frac{\partial\hat{\Psi}_{Em}}{\partial P_{1}}\bigg|_{P_{1}=\tilde{P}_{1}}\!\!\!\!&=&\!\!\!\!\prod_{k}^{}\left[\frac{\psi_{Em}^{(k)}}{\alpha_{k}^{(m)}}\right]^{\alpha_{k}^{(m)}}\left[\frac{\sum_{k}\rho_{Em}^{(k)}}{P_{1}\hat{\Psi}_{Em}(\mathbf{\tilde{p}})}\right]\nonumber\\
   \!\!\!\!&=&\!\!\!\!\left[\hat{\Psi}_{Em}(\mathbf{\tilde{p}})\right]^{\sum_{k}\alpha_{k}^{(m)}}\frac{\sum_{k}\rho_{Em}^{(k)}}{P_{1}\hat{\Psi}_{Em}(\mathbf{\tilde{p}})}\nonumber\\
   \!\!\!\!&=&\!\!\!\!\frac{\sum_{k}\rho_{Em}^{(k)}}{P_{1}} = \frac{\partial\Psi_{Em}}{\partial P_{1}}\bigg|_{P_{1}=\tilde{P}_{1}}
\end{eqnarray}
where $\rho_{Em}^{(k)}$ are the differentiated $\psi_{Em}^{(k)}$'s with respect to $P_{1}$. Similarly, the rest of the partial derivatives in \eqref{eq:partial_diff} can be derived and it can be easily proved to be equal to the partial derivatives of $\Psi_{Em}(\mathbf{\tilde{p}})$, with respect to the corresponding power allocation. Hence, the power allocation obtained through Algorithm A satisfies the KKT conditions of the original optimization problem $\textsc{P}1$. However, it is difficult to analytically prove global optimality. In addition, the geometric programming in Algorithm A can be solved with polynomial time complexity. In order to validate the convergence of the proposed algorithm, simulation results will be provided in Section VII for different sets of wireless channels.
\subsection{Convergence Analysis}
\indent The approximated secrecy rate maximization problem $\textsc{P}3$ given by \eqref{eq:secrecy_rate_max_v3} is convex, and the optimal power allocation $\mathbf{p}^{*}$ can be obtained by solving \eqref{eq:secrecy_rate_max_v3} for a given set of power allocation $\mathbf{\tilde{p}}$. At each iteration, the power allocation $\mathbf{\tilde{p}}$ is updated from the optimal solution $\mathbf{p}^{*}$ determined through the previous iteration. Hence, $\mathbf{\tilde{p}}$ is always a feasible solution of the next iteration, and the optimal power allocation $\mathbf{p}^{*}$ obtained for a given $\mathbf{\tilde{p}}$ will achieve a secrecy rate, which is greater than or equal to that of the previous iteration. This reveals that the achieved secrecy rate will monotonically increase at each iteration, which can be also observed from the simulation results, presented in Fig. \ref{fig:algo_convergence}. Since, the achievable secrecy rate is upper bounded for a given transmit power at the jammers, this algorithm will converge to a solution. Fortunately, the proposed Algorithm A converges to the optimal solution, which is validated through an one-dimensional search, based on bisection and provided in the following section.
\section{Optimality Validation of the secrecy rate maximization algorithm}
\indent In this section, we present an one-dimensional search approach to validate the optimality of the proposed algorithm A. The concept behind this approach is to fix the received total interference power at the destination node and find the optimal power allocation at the jammers\cite{Eldar_Signal_Process_Lett_J06,Gan_CJ_Sig_Process_J11}. The secrecy rate maximization problem $\textsc{P}1$ can be formulated into the following max-min one:
\begin{eqnarray}
 \textsc{P}4:R^{*}=\max_{\mathbf{p}}\min_{t_{i}}~\left(t_{1},t_{2},\ldots,t_{M}\right)&&\nonumber\\
   \textrm{s.t.}~\log_{2}\left(\!\!\frac{1\!+\!\frac{P_{s}|h_{D}|^2}{\sum_{i=1}^{N}P_{i}|g_{D}^{(i)}|^2+\sigma_{D}^{2}}}{1\!+\!\frac{P_{s}|h_{Em}|^2}{\sum_{i=1}^{N}P_{i}|g_{Em}^{(1)}|^{2}+\sigma_{Em}^{2}}}\!\!\right)\!\!\!&\geq&\!\!\!t_{m},\forall m \nonumber\\
   P_{i}\!&\leq&\!\bar{P}{i},\forall i\label{eq:quasi_cvx1}
\end{eqnarray}
where $R^{*}$ is the optimal achieved secrecy rate. By fixing the total received interference (i.e., $\sum_{i=1}^{N}P_{i}|g_{D}^{(i)}|^2$) at the destination to a particular value $t_{0}$, the following subproblem can be formulated as:
\begin{eqnarray}
  \textsc{P}5:~~q^{*}=\max_{\mathbf{p},t}\!\!\!\!&&\!\!\!\! t \nonumber\\
  \textrm{s.t.}\!\!\!\!&& \sum_{i=1}^{N}P_{i}|g_{D}^{(i)}|^2 = t_{0},\nonumber\\
  \!\!\!\!&&\!\!\!\!R_{Em}(t_{0})\!=\!\frac{1+\frac{P_{s}|h_{D}|^2}{t_{0}+\sigma_{D}^{2}}}{1+\frac{P_{s}|h_{Em}|^2}{f_{m}(t_{0})+\sigma_{Em}^{2}}}\!\geq\! t,\forall m\nonumber\\
   \!\!\!&&\!\!\!P_{i} \leq \bar{P}{i},\forall i\label{eq:quasi_cvx2}
\end{eqnarray}
where $f_{m}(t_{0})=\sum_{i=1}^{N}P_{i}|g_{Em}^{(1)}|^{2}$. Next we show that the problem in \eqref{eq:quasi_cvx2} is quasi-convex in terms of $t_{0}$, and therefore, the optimal $t_{0}$ can be obtained through one-dimensional search.\\
\begin{lemma}
$R_{Em}(t_{0})$ is a quasi-concave function in terms of $t_{0}$.
\end{lemma}
\begin{proof}
This can be proved by finding the second derivative of $R_{Em}(t_{0})$ with respect to $t_{0}$ and easily provided that it is negative for any $t_{0}>0$\cite{Eldar_Signal_Process_Lett_J06}.
\end{proof}
\vspace{0.1in}
\hspace{0.1in} In addition, the point-wise infimum of a set of quasi-concave functions is quasi-concave\cite{boyd_B04}. Therefore, the problem $\textsc{P}5$ given by \eqref{eq:quasi_cvx2} is quasi-convex and the optimal power allocation at the jammers can be obtained through Algorithm B.
{\renewcommand\baselinestretch{1}\selectfont
\begin{figure}[t]
    \hspace{0em}\hrulefill

 \hspace{0em} {\bf Algorithm B:} One-Dimensional Search Based on Bisection

     \hspace{0em}\hrulefill

     \hspace{0em} Step 1: Initialize $t_{0}^{(\rm{min})}, t_{0}^{(\rm{max})}$ and $\epsilon$

     \hspace{0em} Step 2: Solve the problem in $\textsc{P}5$ given by \eqref{eq:quasi_cvx2} with $t_{0} = \frac{t_{0}^{(\rm{min})}+3t_{0}^{(\rm{max})}}{4}$.

     \hspace{0em} Step 3: Set $t^{*} = t$.

     \hspace{0em} Step 4: Repeat
\begin{enumerate}
   \item $t_{0} = \frac{t_{0}^{(\rm{min})}+t_{0}^{(\rm{max})}}{2}$.
   \item Solve the problem $\textsc{P}5$ given by \eqref{eq:quasi_cvx2} and obtain the value of $t$
   \item If $t^{*} > t$
   \item \hspace{1em} $t_{0}^{(\rm{min})} = \frac{t_{0}^{(\rm{min})}+t_{0}^{(\rm{max})}}{2} $
   \item else
   \item \hspace{1em} $t_{0}^{(\rm{max})} = \frac{t_{0}^{(\rm{min})}+t_{0}^{(\rm{max})}}{2} $
   \item end
\end{enumerate}
\hspace{-0em} Step 5: Repeat until $t_{0}^{(\rm{max})}-t_{0}^{(\rm{min})}\geq \epsilon$.

\hspace{-1em} \hrulefill
\end{figure}\par}
\section{SOP Analysis over Rayleigh Fading Channels}\label{Sec:SOP_Analysis}
In this section, we analyze the SOP performance of the proposed cooperative jamming scheme over Rayleigh fading channels. In particular, for the system model presented in Sec$.$~II, we assume that $h_{\rm D}$ as well as $g_{\rm D}^{(n)}$ $\forall$ $n=1,2,\ldots,N$ and $\gamma_{{\rm E}_i}$ $\forall$ $i=1,2,\ldots,M$ are standard circularly-symmetric complex Gaussian RVs.

By using the SOP definition of \cite{C:Barros_Secrecy_2006}, the SOP of the proposed cooperative jamming scheme can be obtained as
\begin{equation}\label{Eq:SOP_Definition}
\begin{split}
{\rm P_{out}} =& {\rm Pr}\left[\log_2\frac{\gamma_{\rm D}+1}{\gamma_{\rm E_{max}}+1}<\mathcal{R}\Big| \gamma_{\rm D}>\gamma_{\rm E_{max}}\right]
\\&\times{\rm Pr}\left[\gamma_{\rm D}>\gamma_{\rm E_{max}}\right]+{\rm Pr}\left[\gamma_{\rm D}\leq\gamma_{\rm E_{max}}\right]
\end{split}
\end{equation}
where $\mathcal{R}$ denotes the rate in bits per second (bps) per Hertz. With the utilization of the auxiliary positive real parameter $\mu\triangleq2^\mathcal{R}$ and the negative real parameter $\nu\triangleq2^{-\mathcal{R}}-1$, \eqref{Eq:SOP_Definition} can be rewritten, as shown in Appendix~\ref{App:Derivation_SOP}, as
\begin{equation}\label{Eq:SOP_Elaborations}
\begin{split}
{\rm P_{out}} &= 1-{\rm Pr}\left[\gamma_{\rm E_{max}}<\frac{\gamma_{\rm D}}{\mu}+\nu\right]
\\&= 1 - \mu\int_0^\infty F_{\gamma_{\rm E_{max}}}(x)f_{\gamma_{\rm D}}(\mu x-\mu\nu){\rm d}x.
\end{split}
\end{equation}
\indent In order to solve the integral in \eqref{Eq:SOP_Elaborations}, we first derive a closed-form expression for the PDF of $\gamma_{\rm D}$ as follows. Since $z\triangleq P_s|h_{\rm D}|^2$ is an exponentially distributed RV and $y\triangleq \sum_{n=1}^N P_n|g_{{\rm E}_i}^{(n)}|^2$ is a generalized chi-squared one, by obtaining the CDF of $z$ and the PDF of $y$ by easily  integrating \cite[eq. (2.7)]{B:Sim_Alou_Book} and from \cite[eq. (19)]{J:Mats_GenSumChi} for distinct $P_n$'s, it can be shown that the CDF of $\gamma_{\rm D}$ is given by
\begin{equation}\label{Eq:CDF_gamma_D}
\begin{split}
&F_{\gamma_{\rm D}}(x) = \int_{\sigma_{\rm D}^2}^\infty F_{z}(xw)f_{y}\left(w-\sigma_{\rm D}^2\right){\rm d}w
\\&= 1-\sum_{n=1}^{N}\mathcal{A}_n\exp\left(\frac{\sigma_{\rm D}^2}{P_n}\right)\int_{\sigma_{\rm D}^2}^\infty \exp\left[-\left(\frac{x}{P_s}+\frac{1}{P_n}\right)w\right]{\rm d}w
\\&\stackrel{(a)}{=} 1-P_s\exp\left(-\frac{\sigma_{\rm D}^2x}{P_s}\right)\sum_{n=1}^{N}\frac{\mathcal{A}_nP_n}{P_nx+P_s}
\end{split}
\end{equation}
where $(a)$ follows after using \cite[eq. (3.381/3)]{B:Gra_Ryz_Book} and the definition
\begin{equation}\label{Eq:A_i}
A_n \triangleq \left[P_n\prod_{j=1,j\neq n}^{N}\left(1-\frac{P_j}{P_n}\right)\right]^{-1}.
\end{equation}
By differentiating \eqref{Eq:CDF_gamma_D}, the PDF of $\gamma_{\rm D}$ is easily derived as
\begin{equation}\label{Eq:PDF_gamma_D}
\begin{split}
f_{\gamma_{\rm D}}(x) =& \exp\left(-\frac{\sigma_{\rm D}^2x}{P_s}\right)
\\&\times\sum_{n=1}^{N}\mathcal{A}_nP_n\left[\frac{\sigma_{\rm D}^2}{P_nx+P_s}+\frac{P_sP_n}{\left(P_nx+P_s\right)^2}\right].
\end{split}
\end{equation}
A closed-form expression for the CDF of $\gamma_{\rm E_{max}}$ can be easily obtained using the marginal CDFs of $\gamma_{{\rm E}_i}$ $\forall$ $i$ and the fact that these RVs are independent. In particular, the latter CDFs are derived in closed form similar to the CDF of $\gamma_{\rm D}$ and each is given by \eqref{Eq:CDF_gamma_D} after substituting $\sigma_{\rm D}^2$ with $\sigma_{{\rm E}_i}^2$. Hence, the CDF of $\gamma_{\rm E_{max}}$ can be expressed as
\begin{equation}\label{Eq:CDF_gamma_E_max}
F_{\gamma_{\rm E_{max}}}(x) = \prod_{i=1}^M\left[1-P_s\exp\left(-\frac{\sigma_{{\rm E}_i}^2x}{P_s}\right)\sum_{n=1}^{N}\frac{\mathcal{A}_nP_n}{P_nx+P_s}\right].
\end{equation}

By substituting \eqref{Eq:PDF_gamma_D} and \eqref{Eq:CDF_gamma_E_max} into \eqref{Eq:SOP_Elaborations}, an analytical expression in the form of a single integral for the SOP of the proposed PHY-layer security scheme can be obtained as
\begin{equation}\label{Eq:SOP_Integral_Form}
{\rm P_{out}} = 1 - \mu\exp\left(\frac{\sigma_{\rm D}^2\mu\nu}{P_s}\right)Y
\end{equation}
where integral $Y$ is given by
\begin{equation}\label{Eq:Y_integral}
\begin{split}
Y =& \int_0^\infty \left\{\prod_{i=1}^M\left[1-P_s\exp\left(-\frac{\sigma_{{\rm E}_i}^2x}{P_s}\right)\sum_{n=1}^{N}\frac{\mathcal{A}_n}{x+\lambda_n}\right]\right\}
\\&\times\exp\left(-\xi x\right)
\sum_{n=1}^{N}\frac{\mathcal{A}_n}{\mu}\left[\frac{\sigma_{\rm D}^2}{x-\kappa_n}+\frac{P_s}{\mu\left(x-\kappa_n\right)^2}\right]{\rm d}x
\end{split}
\end{equation}
with $\xi\triangleq P_s^{-1}\sigma_{\rm D}^2\mu$ as well as, for $n=1,2,\ldots,N$, $\kappa_n\triangleq P_s/\left(\mu P_n\right)-\nu$ and $\lambda_n\triangleq P_s/P_n$. By using the closed-form solution for $Y$ included in Appendix~\ref{App:Theorem_1}, a closed-form expression for the SOP of the proposed scheme for arbitrary positive integer values of $N$ and $M$ is given by
\begin{equation}\label{Eq:SOP_Closed_Form_General}
\begin{split}
&{\rm P_{out}} = 1 - \mu\exp\left(\xi\nu\right)\left\{\sum_{n=1}^{N}\frac{\mathcal{A}_n}{\mu}\left[\sigma_{{\rm D}}^2I_{1,0}\left(\xi,\kappa_n,0\right)\right.\right.
\\&\left.\left.+\frac{P_s}{\mu}I_{2,0}\left(\xi,\kappa_n,0\right)\right]+\sum_{\{\alpha_i\}_{i=1}^M}P_s^i \sum_{k_1+k_2+\cdots+k_N=i}\frac{i!}{\prod_{n=1}^{N}k_n!}\right.
\\&\left.\times\left(\prod_{t=1}^{N}\mathcal{A}_t^{k_t}\right)\sum_{n=1}^{N}\frac{\mathcal{A}_n}{\mu}\left[\sigma_{{\rm D}}^2 I_{1,\{k_n\}_{n=1}^N}\left(\psi_i,\kappa_n,\{\lambda_n\}_{n=1}^N\right)\right.\right.
\\&\left.\left.+\frac{P_s}{\mu}I_{2,\{k_n\}_{n=1}^N}\left(\psi_i,\kappa_n,\{\lambda_n\}_{n=1}^N\right)\right]\right\}
\end{split}
\end{equation}
where symbol $\sum_{\{\alpha_i\}_{i=1}^M}$ is used for short-hand representation of the multiple summation $\sum_{i=1}^M\sum_{\alpha_1=1}^{M-i+1}\sum_{\alpha_2=\alpha_1+1}^{M-i+2}\cdots\sum_{\alpha_i=\alpha_{i-1}+1}^{M}$ and the sum $\sum_{k_1+k_2+\cdots+k_N=i}$ is taken over all combinations of nonnegative integer indices $k_1$ through $k_N$ such that the sum of all $k_n$ is $i$. Moreover, $I_{\ell,\{k_n\}_{n=1}^N}\left(\alpha_1,\alpha_2,\left\{\alpha_{3,n}\right\}_{n=1}^{N}\right)$ is given by \eqref{Eq:Basic_Integral_Solution} for $\ell=1,2$ as well as for $k_n$ being positive integer and $\alpha_1$, $\alpha_2$, $\alpha_{3,n} $ $\in\mathbb{R}_+^*$ $\forall$~$n=1,2,\ldots,N$. As an example, for the special case of $N=2$ and $M=1$, the latter SOP expression simplifies to
\begin{equation}\label{Eq:SOP_Closed_Form_M1N2}
\begin{split}
&{\rm P_{out}} = 1 - \mu\exp\left(\xi\nu\right)
\\&\times\left\{\sum_{n=1}^2\frac{\mathcal{A}_n}{\mu}\left[\sigma_{{\rm D}}^2I_{1,0}\left(\xi,\kappa_n,0\right)+\frac{P_s}{\mu}I_{2,0}\left(\xi,\kappa_n,0\right)\right]\right.
\\&-\left.\sum_{n=1}^2\frac{P_s\mathcal{A}_n^2}{\mu}\left[\sigma_{{\rm D}}^2I_{1,1}\left(\psi,\kappa_n,\lambda_n\right)+\frac{P_s}{\mu}I_{2,1}\left(\psi,\kappa_n,\lambda_n\right)\right]\right.
\\&-\left.\frac{P_s\mathcal{A}_1\mathcal{A}_2\sigma_{\rm D}^2}{\mu}\left[I_{1,1}\left(\psi,\kappa_1,\lambda_2\right)+I_{1,1}\left(\psi,\kappa_2,\lambda_1\right)\right]\right.
\\&-\left.\frac{P_s^2\mathcal{A}_1\mathcal{A}_2}{\mu^2}\left[I_{2,1}\left(\psi,\kappa_1,\lambda_2\right)+I_{2,1}\left(\psi,\kappa_2,\lambda_1\right)\right]\right\}
\end{split}
\end{equation}
where $\psi\triangleq \xi + P_s^{-1}\sigma_{{\rm E}_1}^2$,
\begin{subequations}\label{Eq:Is_Special_Case}
\begin{equation}
I_{1,0}\left(\xi,\kappa_n,0\right) = -\exp\left(\xi\kappa_n\right){\rm Ei}\left(-\xi\kappa_n\right),
\end{equation}
\begin{equation}
I_{2,0}\left(\xi,\kappa_n,0\right) = \kappa_n^{-1}+\xi\exp\left(\xi\kappa_n\right){\rm Ei}\left(-\xi\kappa_n\right),
\end{equation}
and
\begin{equation}
I_{1,1}\left(\psi,\kappa_n,\lambda_n\right) = \frac{I_{1,0}\left(\psi,\kappa_n\right)-I_{1,0}\left(\psi,\lambda_n\right)}{\lambda_n-\kappa_n},
\end{equation}
\begin{equation}
I_{2,1}\left(\psi,\kappa_n,\lambda_n\right) = \frac{I_{1,0}\left(\psi,\lambda_n\right)-I_{1,0}\left(\psi,\kappa_n\right)}{\left(\kappa_n-\lambda_n\right)^2}-\frac{I_{2,0}\left(\psi,\kappa_n\right)}{\kappa_n-\lambda_n}.
\end{equation}
\end{subequations}
\begin{table*}[t]
\small
  \centering
\begin{tabular}{!{\VRule[0.7pt]}c!{\VRule[0.7pt]}c|c|c|c!{\VRule[0.7pt]}c|c|c|c!{\VRule[0.7pt]}}
 \specialrule{0.7pt}{0pt}{0pt}
   & \multicolumn{4}{c!{\VRule[0.7pt]}}{
                                    \textbf{ Algorithm B}
                                  } & \multicolumn{4}{c!{\VRule[0.7pt]}}{
                                     \textbf{Algorithm A}
                                  } \\
  \cline{2-9}
  Channels & $P_{1}$ & $P_{2}$ & $P_{3}$ & \begin{tabular}{c}
                                     Achieved\\
                                    Secrecy Rate \\
                                  \end{tabular} & $P_{1}$ & $P_{2}$ & $P_{3}$ & \begin{tabular}{c}
                                     Achieved\\
                                    Secrecy Rate \\
                                  \end{tabular}\\
  \hline
1 & 1.00  & 0 & 0.50  & 1.62 & 1.00 & 0 & 0.50 & 1.62 \\
  \hline
2 & 1.00  & 0 & 0.17  & 2.98 & 1.00 & 0 & 0.16 & 2.98 \\
  \hline
3 & 0.47  & 0 & 0.35  & 1.68 & 0.46 & 0 & 0.34 & 1.68 \\
  \hline
4 & 1.00  & 0.43 & 0  & 2.72 & 1.00 & 0.42 & 0  & 2.73 \\
  \hline
5 & 0  & 0.28 & 0.31  & 1.09 & 0 & 0.28 & 0.31 & 1.09\\
 \specialrule{0.7pt}{0pt}{0pt}
\end{tabular}
\label{Table:power_allocations}
\caption{ The optimal power allocation at the jammers based on Algorithm A and Algorithm B, for different sets of wireless channels.}
\end{table*}
\section{Numerical Results and Discussions}
\indent In order to validate the performance of the proposed algorithms, we consider the secrecy network shown in Fig. 1, with a source-destination pair, three ($N=3$) cooperative jammers and two ($M=2$) eavesdroppers. In the following simulations, all the channel coefficients involved are generated using zero-mean circularly symmetric independent and identically distributed complex Gaussian RVs. In addition, the noise variances at the destination and the eavesdroppers are assumed to be 0.1.\\
\indent To assess the convergence of the proposed secrecy rate maximization algorithm, the available maximum transmit powers at the source and relays have been set to, $P_{s} =2$, $P_{1} =1$, $P_{2} =1$ and $P_{3} =3$. Fig. \ref{fig:algo_convergence} depicts the convergence of the achievable secrecy rates for a set of different feasible channels. As it is evident from this figure, the proposed algorithm converges, while the achievable secrecy rates increase with the iteration number. In addition, it has been observed that the proposed Algorithm A converges to the same secrecy rate, with different initialization of transmit powers at the jammers. However, we could not provide analytical results to prove this convergence. As we discussed in the convergence analysis of the algorithm, it can be observed that the achievable secrecy rate monotonically increases with the iteration number.\\
\indent Next, we compare the performance of the proposed algorithm with the existing scheme in \cite{Gan_CJ_Sig_Process_J11} and the best jammer selection scheme. The cooperative jamming scheme in \cite{Gan_CJ_Sig_Process_J11} has been developed using both convex optimization approach and one dimensional search scheme in the presence of a single eavesdropper whereas the best jammer is selected from available cooperative jammers in the best jammer selection scheme. In order to evaluate this comparison, the same secrecy network in the previous simulation is considered with a single eavesdropper and with the same noise variance 0.1 at all the nodes. Fig.~\ref{fig:perfor_comparision} depicts the achieved secrecy rates for different available transmit power at the source and the cooperative jammers for different sets of channels, where it is assumed that the maximum available transmit power at the source and the cooperative jammers are the same. As seen in Fig.~\ref{fig:perfor_comparision}, both the proposed algorithm and the scheme in \cite{Gan_CJ_Sig_Process_J11} achieve the same secrecy rates for different sets of channels with the same transmit power constraints and better secrecy rates than the best jammer selection scheme. This confirms that the proposed algorithm shows the same performance as the optimal scheme in \cite{Gan_CJ_Sig_Process_J11} and outperforms the best jammer selection scheme.\\
\indent Next, we evaluate the optimality of the power allocation obtained through the proposed Algorithm  A. In order to do this, we simulate Algorithm  B for the same set of channels considered for Algorithm A. Table I presents the power allocation and the secrecy rates obtained through Algorithm B that is based on one-dimensional search and on the Algorithm A. As we can conclude from this table, the power allocation and achieved secrecy rates are identical for different sets of channels in both algorithms. Note that there are small differences in the power allocation and achieved secrecy rates, due to the accuracy or precision of software used. However, these results provided in Table I confirm the optimality of the proposed secrecy rate maximization Algorithm A.\\
\indent By numerically evaluating \eqref{Eq:SOP_Closed_Form_M1N2}, Fig.~\ref{Fig:SOP_M1N2_vsRate} depicts SOP performance as a function of rate $\mathcal{R}$ in bps per Hertz for $N=2$ cooperative jammers, $M=1$ eavesdropper and various power levels. It is shown in this figure that computer simulation results for SOP match perfectly with the equivalent numerical ones, for all considered parameters. As expected, SOP degrades with increasing values for $\mathcal{R}$. In addition, as the transmit power of the source $S$ increases and the transmit powers at the two cooperative jammers ${\rm J}_1$ and ${\rm J}_2$ decrease, SOP improves. The best SOP performance in this figure for all considered $\mathcal{R}$ values is achieved with $P_s=15$ dB, $P_1=0$ dB and $P_2=2$ dB, and the lower value for SOP is $0.5$.\\
\indent The SOP performance as a function of source $S$'s transmit power $P_s$ is illustrated in Fig.~\ref{Fig:SOP_variousMN_vsPs}. The following transmission scenarios have been considered: \textit{i}) Scenario $1$: $\mathcal{R}=1$, $P_1=-4$ dB and $P_i=P_1+(i-1)$ dB with $i=2,3$ and $4$; and \textit{ii}) Scenario $2$:  $\mathcal{R}=0.01$, $P_1=1$ dB and $P_i=P_1+(i-1)$ dB with $i=2,3$ and $4$. For the SOP results, the single-integral expression given by \eqref{Eq:SOP_Integral_Form} after substituting \eqref{Eq:Y_integral} and the closed-form expression given by
\eqref{Eq:SOP_Closed_Form_General}  for arbitrary values of $N$ and $M$ as well as the closed-form expression given by \eqref{Eq:SOP_Closed_Form_M1N2} for $M=2$ and $N=1$ have been numerically evaluated. As clearly shown, computer simulation results for SOP coincide with the numerical ones, for all considered parameters. Furthermore, it is evident that, for the same values of $N$ and $M$, the SOP performance of Scenario $2$ is always better than that of Scenario $1$. In both scenarios, the minimum SOP is accomplished with $N=M=1$ and the maximum with $N=M=4$. Also, as expected, SOP improves with increasing values of $P_s$ for all considered cases. In addition, it is shown in this figure that, as $M$ increases while $N$ is kept constant, SOP degrades significantly. This performance degradation can be confronted for some range of $P_s$ values by increasing $N$. However, increasing $N$ introduces a SOP performance penalty, that needs to be taken under consideration when designing a cooperating jamming scheme.
\begin{figure}[!t]
\centering \includegraphics[scale = 0.58]{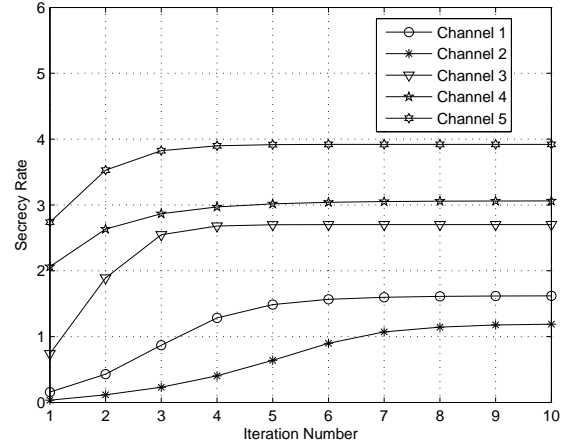}
\caption{The convergence of the proposed secrecy rate maximization Algorithm A, for different sets of wireless channels.}\label{fig:algo_convergence}
\end{figure}
\begin{figure}[!t]
\centering \includegraphics[scale = 0.62]{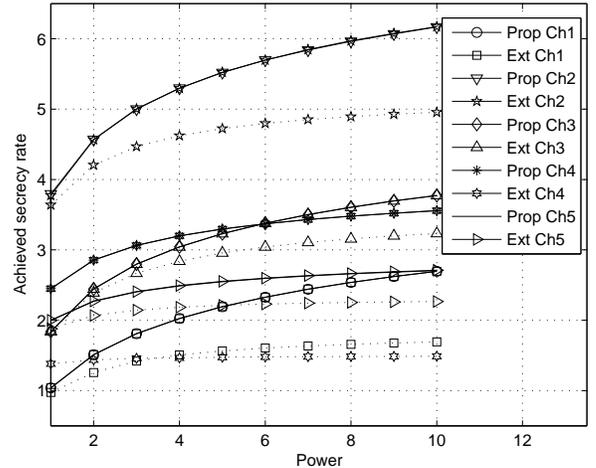}
\caption{The achieved secrecy rates of Algorithm A, the scheme in \cite{Gan_CJ_Sig_Process_J11} and best jammer selection scheme for five sets of different wireless channels with different maximum available transmit power. The dotted lines denote the best jammer selection scheme.}\label{fig:perfor_comparision}
\end{figure}
\begin{figure}[!t]
\centering
\includegraphics[scale = 0.45]{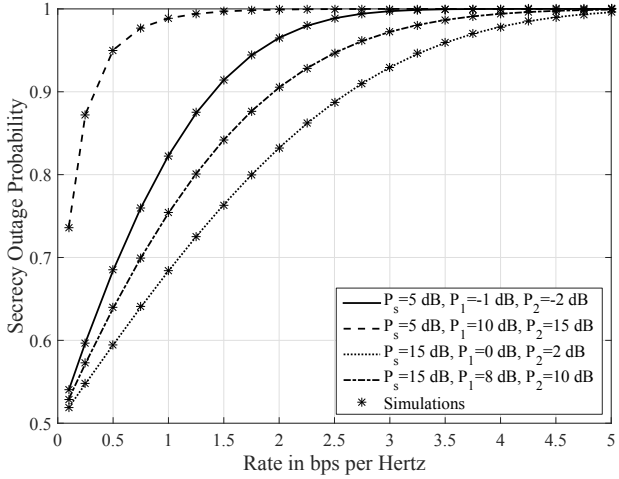}
\caption{${\rm P_{out}}$, as a function of the rate, $\mathcal{R}$, in bps per Hz, for $N=2$ cooperative jammers, $M=1$ eavesdroppers and various power levels.}
\label{Fig:SOP_M1N2_vsRate}
\end{figure}
\begin{figure}[!t]
\centering
\includegraphics[scale = 0.45]{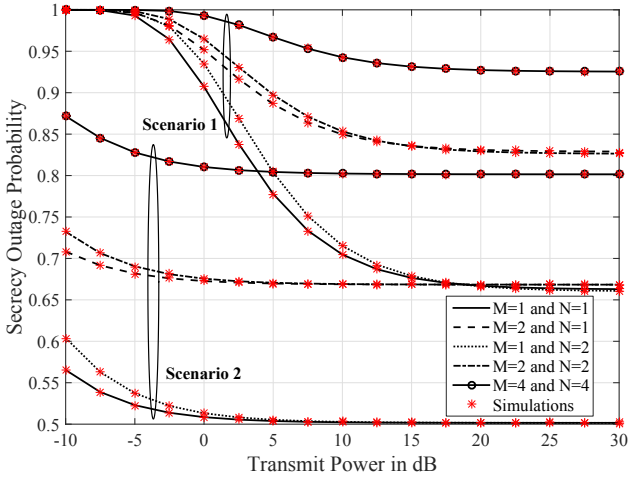}
\caption{SOP performance, ${\rm P_{out}}$, as a function of source's transmit power, $P_s$, in dB for both Scenarios $1$ and $2$ as well as various numbers of cooperative jammers and eavesdroppers.}
\label{Fig:SOP_variousMN_vsPs}
\end{figure}
\section{Conclusions}
\indent In this paper, we studied the power allocation problem of secrecy rate maximization with cooperative jammers, in the presence of multiple eavesdroppers. For this problem, a feasibility condition was first derived for power allocation in order to achieve positive secrecy rate. Then, the original non-convex secrecy rate maximization problem was solved to obtain the optimal power allocation at the jammers. The proposed optimal iterative approach was developed by approximating the secrecy rate function and formulating the corresponding problem into a geometric programming problem for a given set of power allocation at the jammers. In order to validate the optimality of the developed algorithm, we also developed an one-dimensional search algorithm based on bisection. In addition, the SOP analysis of the proposed cooperative jamming approach was derived for Rayleigh fading channels. Simulation results were provided to validate the optimality and convergence of the proposed algorithm as well as the theoretical derivation of SOP analysis. These results confirm that the proposed algorithm yields the optimal power allocation at the jammers, whereas the numerical simulation results demonstrate the correctness of theoretical derivations of the SOP analysis.
%---------------------------------------------------------------------------------------------------------------------------
%---------------------------------------------------------------------------------------------------------------------------
\appendices
\renewcommand{\theequation}{A.\arabic{equation}}
\setcounter{equation}{0}
\section{Proof of Lemma 1}\label{App:Lemma1}
Function $g(\mathbf{x})$ can be written as
\begin{eqnarray}
   g(\mathbf{x})\!\!\!\!\!&=&\!\!\!\!\!\sum_{k=1}^{K}a_{k}\left[\frac{w_{k}(\mathbf{x})}{a_{k}} \right]\geq\prod_{k=1}^{K}\! \left[\frac{w_{k}(\mathbf{x})}{a_{k}}\right]^{a_{k}}\!\!\!=\!\hat{g}(\mathbf{\hat{x}})\label{eq:GM_AM_inequality}
\end{eqnarray}
where the inequality in \eqref{eq:GM_AM_inequality} is obtained from the arithmetic-geometric mean inequality. This inequality holds with equality when $a_{k}=\frac{w_{k}(\mathbf{\hat{x}})}{g(\mathbf{\hat{x}})}$ as follows:
\begin{eqnarray}
  \hat{g}(\mathbf{\hat{x}})\!\!\!&=&\!\!\!\prod_{k=1}^{K} \left[\frac{w_{k}(\mathbf{\hat{x}})}{\bar{a}_{k}}\right]^{\bar{a}_{k}}\nonumber\\
   \!\!\!&=&\!\!\!\prod_{k=1}^{K} g(\mathbf{\hat{x}})^{\sum_{k=1}^{K}\frac{w_{k}(\mathbf{\hat{x}})}{g(\mathbf{\hat{x}})}} = g(\mathbf{\hat{x}})\nonumber
\end{eqnarray}
where
\begin{eqnarray}
   \bar{a}_{k}\!\!\!&=&\!\!\!\frac{w_{k}(\mathbf{\hat{x}})}{g(\mathbf{\hat{x}})}~\textrm{and}~\sum_{k=1}^{K}\frac{w_{k}(\mathbf{\hat{x}})}{g(\mathbf{\hat{x}})} = 1.
\end{eqnarray}
This completes the proof of Lemma 1.
\renewcommand{\theequation}{B.\arabic{equation}}
\setcounter{equation}{0}
\section{Derivation of \eqref{Eq:SOP_Elaborations}}\label{App:Derivation_SOP}
Starting from \eqref{Eq:SOP_Definition} and using the definition of conditional probability results in
\begin{equation}\label{Eq:SOP_Initial_Elaborations}
\begin{split}
{\rm P_{out}} &= 1+{\rm Pr}\left[\frac{\gamma_{\rm D}}{\mu}+\nu<\gamma_{\rm E_{max}}<\gamma_{\rm D}\right]-{\rm Pr}\left[\gamma_{\rm E_{max}}<\gamma_{\rm D}\right]
\\&=1-\int_0^\infty F_{\gamma_{\rm E_{max}}}\left(\frac{y}{\mu}+\nu\right)f_{\gamma_{\rm D}}(y){\rm d}y.
\end{split}
\end{equation}
By using the change of variables $x\rightarrow y/\mu+\nu$ and the fact that $\nu<0$, yields \eqref{Eq:SOP_Elaborations}.
%---------------------------------------------------------------------------------------------------------------------------
%---------------------------------------------------------------------------------------------------------------------------

%---------------------------------------------------------------------------------------------------------------------------
%---------------------------------------------------------------------------------------------------------------------------
\renewcommand{\theequation}{C.\arabic{equation}}
\setcounter{equation}{0}
\section{Closed-Form Solution for \eqref{Eq:Y_integral}}\label{App:Theorem_1}
To solve integral $Y$ given by \eqref{Eq:Y_integral} that appears in the SOP expression given by \eqref{Eq:SOP_Integral_Form}, we first make use of the multinomial expansion \cite[eq$.$ (23)]{J:Alexandg_Comparative} for the $M$-factor product, yielding
\begin{equation}\label{Eq:Multinomial_Expansion}
\begin{split}
&\prod_{i=1}^M\left[1-P_s\exp\left(-\frac{\sigma_{{\rm E}_i}^2x}{P_s}\right)\sum_{n=1}^{N}\frac{\mathcal{A}_n}{x+\lambda_n}\right]
\\& = 1 + \sum_{\{\alpha_i\}_{i=1}^M}P_s^i\exp\left(-\frac{x}{P_s}\sum_{j=1}^{i}\sigma_{{\rm E}_{\alpha_j}}^2\right)\left(\sum_{n=1}^{N}\frac{\mathcal{A}_n}{x+\lambda_n}\right)^i.
\end{split}
\end{equation}
Then, in the latter expression, we utilize the multinomial theorem to expand the $i^{\rm{th}}$ power of the $N$-term sum as follows
\begin{equation}\label{Eq:Multinomial_Theorem}
\begin{split}
\left(\sum_{n=1}^{N}\frac{\mathcal{A}_n}{x+\lambda_n}\right)^i = \sum_{k_1+k_2+\cdots+k_N=i}&\frac{i!}{\prod_{n=1}^{N}k_n!}
\\&\times\prod_{t=1}^{N}\frac{\mathcal{A}_t^{k_t}}{(x+\lambda_t)^{k_t}}.
\end{split}
\end{equation}
Substituting \eqref{Eq:Multinomial_Theorem} into \eqref{Eq:Multinomial_Expansion} and then into \eqref{Eq:Y_integral}, integral $Y$ can be rewritten as
\begin{equation}\label{Eq:Y_integral_Expanded}
\begin{split}
 Y =& \sum_{n=1}^{N}\frac{\mathcal{A}_n}{\mu}\left[\sigma_{{\rm D}}^2I_{1,0}\left(\xi,\kappa_n,0\right)\frac{P_s}{\mu}+I_{2,0}\left(\xi,\kappa_n,0\right)\right]
\\&+\sum_{\{\alpha_i\}_{i=1}^M}P_s^i \sum_{k_1+k_2+\cdots+k_N=i}\frac{i!\prod_{t=1}^{N}\mathcal{A}_t^{k_t}}{\prod_{n=1}^{N}k_n!}
\\&\times\sum_{n=1}^{N}\frac{\mathcal{A}_n}{\mu}\left[\sigma_{{\rm D}}^2 I_{1,\{k_n\}_{n=1}^N}\left(\psi_i,\kappa_n,\{\lambda_n\}_{n=1}^N\right)\right.
 \\&\left.+\frac{P_s}{\mu}I_{2,\{k_n\}_{n=1}^N}\left(\psi_i,\kappa_n,\{\lambda_n\}_{n=1}^N\right)\right]
\end{split}
\end{equation}
where $\psi_i\triangleq\xi+P_s^{-1}\sum_{j=1}^{i}\sigma_{{\rm E}_{\alpha_j}}^2$. In addition, $I_{\ell,\{k_n\}_{n=1}^N}\left(\alpha_1,\alpha_2,\left\{\alpha_{3,n}\right\}_{n=1}^{N}\right)$ for $\ell=1,2$ as well as for $k_n$ being positive integer and $\alpha_1$, $\alpha_2$, $\alpha_{3,n} $ $\in\mathbb{R}_+^*$ $\forall$~$n=1,2,\ldots,N$ is defined as
\begin{equation}\label{Eq:Basic_Integral}
\begin{split}
&I_{\ell,\{k_n\}_{n=1}^N}\left(\alpha_1,\alpha_2,\left\{\alpha_{3,n}\right\}_{n=1}^{N}\right)
\\&= \int_0^\infty \frac{\exp\left(-\alpha_1 x\right)}{\left(x+\alpha_2\right)^\ell\prod_{n=1}^N\left(x+\alpha_{3,n}\right)^{k_n}}{\rm d}x.
\end{split}
\end{equation}
By using \cite[Sec$.$~2.1]{B:Gra_Ryz_Book} for the rational function integrand in \eqref{Eq:Basic_Integral} in order to rewrite the integral as summations of integrals, it can be shown that
\begin{equation}\label{Eq:Basic_Integral}
\begin{split}
&I_{\ell,\{k_n\}_{n=1}^N}\left(\alpha_1,\alpha_2,\left\{\alpha_{3,n}\right\}_{n=1}^{N}\right) = \sum_{i=1}^\ell Z_i\int_0^\infty\frac{\exp\left(-\alpha_1 x\right)}{\left(x+\alpha_2\right)^i}{\rm d}x
\\&+ \sum_{j=1}^N \sum_{i=1}^{k_j} \Theta_i^{(k_j)}\int_0^\infty\frac{\exp\left(-\alpha_1 x\right)}{\left(x+\alpha_{3,j}\right)^i}{\rm d}x
\end{split}
\end{equation}
where the real-valued parameter $Z_i$ is given by
\begin{equation}\label{Eq:Z}
Z_{\ell-k+1} = \frac{1}{(k-1)!}\frac{{\rm d}^{k-1}}{{\rm d}x^{k-1}} \zeta(x)\Big|_{x=-\alpha_2}
\end{equation}
for $k\leq\ell$ with $\zeta(x)=\prod_{n=1}^N \left(x+\lambda_n\right)^{-k_n}$, and the real-valued parameter $\Theta_i^{(k_j)}$ by
\begin{equation}\label{Eq:Theta}
\Theta_{k_j-k+1}^{(k_j)} = \frac{1}{(k-1)!}\frac{{\rm d}^{k-1}}{{\rm d}x^{k-1}} \theta_j(x)\Big|_{x=-\lambda_j}
\end{equation}
for $k\leq k_j$ with $\theta_j(x)=\left(x+\alpha_2\right)^{-1}\prod_{n\neq j}^N \left(x+\lambda_n\right)^{-k_n}$. By making use of \cite[eq. (3.353/2)]{B:Gra_Ryz_Book} for the integrals appearing in \eqref{Eq:Basic_Integral}, yields
\begin{equation}\label{Eq:Basic_Integral_Solution}
\begin{split}
&I_{\ell,\{k_n\}_{n=1}^N}\left(\alpha_1,\alpha_2,\left\{\alpha_{3,n}\right\}_{n=1}^{N}\right) = \sum_{i=1}^\ell \frac{Z_i}{(i-1)!}\sum_{r=1}^{i-1}(r-1)!
\\&\times\left(-\alpha_1\right)^{i-r-1}\alpha_2^{-r}-\frac{\left(-\alpha_1\right)^{i-1}}{(i-1)!}\exp\left(\alpha_1\alpha_2\right){\rm Ei}\left(-\alpha_1\alpha_2\right)
\\&+ \sum_{j=1}^N \sum_{i=1}^{k_j} \frac{\Theta_i^{(k_j)}}{(i-1)!}\sum_{r=1}^{i-1}(r-1)!\left(-\alpha_1\right)^{i-r-1}\alpha_{3,j}^{-r}-\frac{\left(-\alpha_1\right)^{i-1}}{(i-1)!}
\\&\times\exp\left(\alpha_1\alpha_{3,j}\right){\rm Ei}\left(-\alpha_1\alpha_{3,j}\right).
\end{split}
\end{equation}
Finally, by replacing \eqref{Eq:Basic_Integral_Solution} in \eqref{Eq:Y_integral_Expanded} yields a closed-form expression for integral $Y$.
\bibliographystyle{ieeetr}
\linespread{1}
\bibliography{Myjournals,myrefs,mybooks}
\balance
\end{document}